\documentclass[english]{article}

\usepackage{custom_tex}

\title{\Large Conditional Quasi-Monte Carlo with Constrained Active Subspaces}
\date{\today}
\author{Sifan Liu\thanks{Email: \texttt{sfliu@stanford.edu}}} 
\affil{Department of Statistics, Stanford University}

\begin{document}

\maketitle

\begin{abstract}
Conditional Monte Carlo or pre-integration is a powerful tool for reducing variance and improving the regularity of integrands when using Monte Carlo and quasi-Monte Carlo (QMC) methods. To select the variable to pre-integrate, one must consider both the variable's importance and the tractability of the conditional expectation. For integrals over a Gaussian distribution, any linear combination of variables can potentially be pre-integrated. Liu and Owen \cite{liu2022pre} propose to select the linear combination based on an active subspace decomposition of the integrand. However, pre-integrating the selected direction might be intractable. In this work, we address this issue by finding the active subspace subject to constraints such that pre-integration can be easily carried out. The proposed algorithm also provides a computationally-efficient alternative to dimension reduction for pre-integrated functions. The method is applied to examples \blue{from computational finance, density estimation, and computational chemistry}, and is shown to achieve smaller errors than previous methods. 
\end{abstract}

\section{Introduction}

Quasi-Monte Carlo (QMC) and its randomized versions \sloppy{(RQMC)} achieve smaller integration errors than plain Monte Carlo (MC), in both theory and practice. However, unlike Monte Carlo, the performance of QMC is highly dependent on the dimension and smoothness of the integrand. 
Directly applying QMC without properly addressing the high dimensionality or non-smoothness of the integrand may not lead to a significant improvement over plain Monte Carlo. To take full advantage of QMC, a dimension reduction or smoothing procedure is often necessary.

Dimension reduction in QMC refers to the process of rearranging the integration variables so that the important ones are integrated with the first few dimensions of QMC points, which takes advantage of the fact that QMC points are typically more uniformly distributed in the first few dimensions \cite{moro:cafl:1994}. 
For integrals $\EE{f(\bfx)}$ over a Gaussian random vector $\bfx\sim\N(0,\Sigma)$, dimension reduction is often achieved by finding an orthogonal matrix to rotate $\bfx$, or equivalently,
choosing a square-root $R$ of the covariance matrix s.t. $RR\tran=\Sigma$, and 
using the estimator
\[
\frac1n\sum_{i=1}^n f(R\Phi^{-1}(\bfu_i)),
\]
where $\bfu_i\in[0,1]^d$ are the QMC samples and $\Phi^{-1}$ is the inverse cumulative distribution function (CDF) of the standard normal distribution applied to each coordinate. 
For example, $R$ can be obtained by the Cholesky decomposition or eigendecomposition of $\Sigma$. 
For many financial applications, the randomness is driven by the Brownian motion, which allows for the Brownian bridge (BB) construction of the square-root $R$. In this case, the Cholesky decomposition is known as the standard construction and the eigendecomposition is known as the PCA construction of the Brownian motion. 
For payoff functions and their derivatives related to Asian options, PCA and BB constructions are usually superior to the standard construction as they result in a lower effective dimension. 
%
These choices of $R$ do not take into account any information of the integrand, and as a result, their performance can vary greatly depending on the integrand. Over the years, various integrand-aware dimension reduction methods have been developed to address this issue.
Imai and Tan \cite{imai2006general} propose a method called linear transformation (LT), which aims to minimize the effective dimension of a first-order Taylor approximation of the integrand expanded at some predetermined points. 
However, the predetermined points might not be representative of the whole function. For instance, for an integrand like $(f(\bfx)-K)_+:=(f(\bfx)-K)\Indc{f(\bfx)-K\geq0}$, it is possible that its gradient is all-zero at some pre-selected points.
To address this issue, Weng et al. \cite{weng2017efficient} propose to apply k-means to cluster the gradients of the integrand at some randomly sampled points. They then determine the rotation matrix by the QR decomposition of clusters' centroids.
More recently, Xiao and Wang \cite{xiao2019enhancing} propose using the principal components of the gradients, known as gradient PCA (GPCA), to find the orthogonal matrix. The GPCA method can be applied in companion with importance sampling to reduce effective dimension of the resulting integrands \cite{zhan:wang:he:2021}.

To improve smoothness, there are three classes of methods. The first class of methods aligns the discontinuities with axes because axes-parallel discontinuities are QMC-friendly \cite{wang2013pricing}. For example, when pricing some digital options, the standard construction can outperform the PCA and BB constructions \cite{papa:2002}, which is because the discontinuities are axes-parallel under the standard construction.
If the integrand takes the form of $f(\bfx)\Indc{w\tran \bfx\geq K}$ where $f$ is smooth and $\|w\|_2=1$ for $\bfx\sim\N(0,I)$, one can make the discontinuity QMC-friendly by choosing an orthogonal matrix $U$ such that the first column of $U$ is $w$. With the change-of-variable $\bfx=U\bfz$, the integrand becomes $f(U\bfz) \Indc{z_1\geq K}$, whose discontinuities are parallel to the second through the last axes. This is the orthogonal transformation (OT) method proposed by \cite{wang2013pricing}.
He and Wang \cite{he2014good} generalize the OT method for integrands with discontinuities in more than one direction. For integrands of the form $f(\bfx)\Indc{W\tran \bfx\geq K}$, where $W\in\R^{r\times d}$, they propose to choose the rotation matrix by the QR decomposition of $W$.
Imai and Tan \cite{imai2014pricing} propose a method called integrated linear transformation (ILT) that combines the LT and OT methods.
The second class of methods remove the discontinuities completely by utilizing the separation-of-variable property \cite{wang2016handling}. If the discontinuities are caused by the indicator function \sloppy{$\Indc{\alpha(\bfx_{2:d})\leq x_1 \leq \beta(\bfx_{2:d}) }$}, then the discontinuity can be avoided by sampling $x_1$ directly from the truncated normal distribution 
$\N(0,1)\indc{[\alpha(\bfx_{2:d},\beta(\bfx_{2:d}]}$. 
The third class of methods is known as conditional Monte Carlo (CMC) or pre-integration, where one or more variables are integrated out by a closed form. Griebel et al. \cite{griebel2010smoothing} \cite{griebel2013smoothing} and Griewank et al. \cite{grie:kuo:leov:sloa:2018} show that for integrands of the form $f(\bfx)\Indc{\phi(\bfx)\geq 0}$, integrating out some variable $x_j$ results in unlimited smoothness under the key assumption that $\phi$ is smooth and monotone in $x_j$ with all other variables held fixed.
Recently, Gilbert et al. \cite{gilb:kuo:sloa:2021:tr} show that the monotonicity condition is indeed necessary for pre-integration to bring smoothness. 
On the practical side, the monotonicity condition is essential in order for the pre-integration step to have a closed form in many option pricing tasks. Hence, for sake of the smoothing effect and tractability of pre-integration, it is important to choose a variable in which the monotonicity condition holds. 

However, the monotonicity condition alone is not enough for pre-integration to bring a significant improvement. Liu and Owen \cite{liu2022pre} observe that for some Asian options, pre-integrating $x_1$ under the standard construction of the Brownian motion does not reduce the RQMC variance by much. This is because the first variable does not explain much of the variance in the integrand. 
Therefore, one must also consider the variable's importance when choosing the pre-integrating variable.
Moreover, when the random variables are Gaussian, which is the case in many financial applications, we can pre-integrate over any linear combination of the input variables. The problem then boils down to choosing an orthogonal matrix $U$ and replace the integrand $f(\bfx)$ by $f_U(\bfx)=f(U\bfx)$, and the pre-integrated variable is $x_1$ in the rotated integrand $f_U$.
Liu and Owen \cite{liu2022pre} propose a systematic way to choose the rotation matrix $U$. The method is motivated by the active subspace \cite{cons:2015}, which is the subspace spanned by the top eigenvectors of the matrix
\begin{align}
     C=\EE{\nabla f(\bfx) \nabla f(\bfx)\tran}.
     \label{def: C}
\end{align}
The intuition is that the top eigenvectors of $C$ correspond to the directions where the function $f$ has the largest variation on average. By taking $U$ to be the eigenvectors of $C$ in descending order, the first variable of $f_U$ corresponds to the direction where $f_U$ has the greatest variance. So pre-integrating $x_1$ can potentially reduce a large proportion of variance. Additionally, the remaining $d-1$ variables are also sorted in decreasing order of importance, which reduces the effective dimension of the resulting integrand.  

One limitation of this approach is that the conditional expectation $\tf_U(\bfx_{2:d})=\EE{f_U(\bfx)\mid \bfx_{2:d}}$ might not have a closed form. Particularly, the rotated integrand $f_U$ might not be monotone in $x_1$. Although in the original paper \cite{liu2022pre}, it is shown that $f_U$ is monotone in $x_1$ when $f$ is the payoff function of the Asian call option or some basket options, this is not always the case. For example, this method fails when the underlying assets of the basket option have negative correlation or when the asset prices do not follow the simple Black-Scholes model, as the pre-integration step does not have a nice closed form.

To address this issue, this paper proposes a method that takes into account all three factors when applying conditional QMC:
\begin{itemize}
     \item Variable importance
     \item Tractability of the conditional expectation
     \item Dimension reduction for the remaining variables
\end{itemize}
in a unified way. The main idea is to choose the pre-integration direction in a similar manner as the active subspace method, but with certain constraints so that pre-integration can be conducted easily. We refer to this proposed algorithm as the constrained active subspaces (CAS) method. 
The method is flexible enough to incorporate a variety of constraints. For example, under the stochastic volatility models, we impose the constraints so that the volatility does not depend on the pre-integration variable, which facilitates a closed form for the pre-integration step. 
In some cases, if we have prior knowledge about the important direction of an integrand, we can also fix the pre-integration direction to that specific direction. 
Additionally, the CAS method is easy to implement as it only involves PCA of gradients of the original integrand at random points. The proposed method is applicable to a broad class of derivative pricing tasks, including \blue{pricing options involving multiple underlying assets}, Greeks, and pricing under stochastic volatility models. \blue{Moreover, its scope extends beyond computational finance, as it finds applications in diverse fields, such as conditional density estimation \cite{l2022monte} and simulating chemical reaction networks.} Our numerical results show that the proposed method improves upon previous methods. 

The rest of the paper is organized as follows. In Section~\ref{sec: background}, we introduce some backgrounds on QMC, conditional Monte Carlo, and derivative pricing. 
Next, we discuss how to select the orthogonal matrix $U$ in two steps: (1) selecting the first column and (2) selecting the remaining columns based on the first column. 
The method for step (2) is more general while the strategies for step (1) are specific to the problem at hand. Therefore, we first focus on step (2) in Section \ref{sec: AS constraints} and then discuss step (1) in Section \ref{sec: examples}.
In Section~\ref{sec: AS constraints}, we introduce the idea of constrained active subspaces (CAS), which allows us to determine the second through last columns of $U$ given the first column.
Section~\ref{sec: examples} addresses the issue of selecting the first column of $U$ for some specific examples, including option pricing under stochastic volatility models and evaluating some Greeks.
Section~\ref{sec: experiments} presents numerical results on pricing a spread option, call options under stochastic volatility models (including Hull-White, Heston, Stein-Stein models), Greeks, \blue{as well as examples of conditional density estimation and chemical reaction networks}. Section~\ref{sec: conclusion} has our conclusions.

\textbf{Notations. } For a positive integer $d$, denote $1:d=\{1,2,\ldots,d\}$. If $u$ is a subset of $1:d$, denote $-u$ as $1:d\setminus u$. For an integer $j\in1:d$, we use $-j$ to represent $1:d\setminus\{j\}$ when the context is clear. \blue{Let $|u|$ denote the cardinality of the set $u$.} For $\bfx\blue{\in\R^d}$, let $\bfx_u\in\R^{|u|}$ be the vector containing only the $x_j$ if $j\in u$. For a matrix $U\in\R^{d\times d}$, we denote $U_{\cdot,u}$ as the submatrix of $U$ containing the columns of indices belong to $u$. We use $\|\bfx\|$ to denote the $L_2$ norm of a vector $\bfx$. We always use the notation $\bfx$ or $\bfz$ to denote a Gaussian random vector unless otherwise specified. For an eigendecomposition, we always sort the eigenvalues in descending order along with their corresponding eigenvectors.

\section{Background}
\label{sec: background}

In this section, we provide some background on randomized QMC (RQMC), conditional Monte Carlo (CMC), and derivative pricing in financial engineering.

\subsection{RQMC}

QMC is an alternative of Monte Carlo for estimating an integral $\mu=\int_{[0,1]^d} f(\bfx)\rd \bfx$. The estimator takes the form of
\begin{align*}
     \hat\mu=\frac1n\sum_{i=1}^n f(\bfu_i),
\end{align*}
where the points $\bfu_i\in[0,1]^d$ are chosen deterministically rather than i.i.d. uniformly as in plain MC.
The QMC point sets are designed to have low discrepancy, which result in smaller integration error than the probabilistic MC error of $O(n^{-1/2})$. Specifically, for integrands of bounded Hardy-Krause variation, QMC yields a deterministic error bound of $O(n^{-1}(\log n)^d)$ \cite{nied:1992}. 

One construction of QMC is based on digital nets. We call $E(\bfk,\bfc)=\prod_{j=1}^d \left[\frac{c_j}{b^{k_j}}, \frac{c_j+1}{b^{k_j}} \right)$ \blue{an elementary interval in base $b$ for integers $k_j\geq 0$ and $c_j\in\{0,1,\ldots,b^{k_j}-1\}$}. This interval has volume $b^{-|\bfk|}$ where $|\bfk|=\sum_{j=1}^d k_j$. A digital net consisting of $b^m$ points is called a $(t,m,d)$-net in base $b$ if all elementary intervals $E(\bfk,\bfc)$ with volume $b^{-|\bfk|}$ no smaller than $b^{t-m}$ contain exactly the volume fraction of points, i.e. $2^{m-|\bfk|}$ points. The parameter $t$, known as the quality parameter, is nondecreasing in $d$ and smaller is better.
An infinite sequence $\{\bfx_i\}_{i\geq 0}$ is called a $(t,d)$-sequence if for all $k\geq 0$ and $m\geq t$ the point set $\{\bfx_{kb^m+i },\, 0\leq i< b^m \}$ is a $(t,m,d)$-net. The most commonly used nets and sequences are those of Sobol' \cite{sobo:1967:tran}. Sobol' sequences are $(t,d)$-sequences in base 2. 

The QMC points can also be randomized so that each point $\bfx_i$ follows $\unif([0,1]^d)$ individually while collectively they have low discrepancy.
Using randomized QMC (RQMC), the intgration error can be conveniently estimated by the standard error of independent random replicates.
One example of RQMC is the nested scrambled digital net introduced in \cite{rtms}. 
\blue{In one-dimension, nested uniform scrambling begins by dividing the unit interval into two halves, $[0,1/2)$ and $[1/2,1)$. The points in the two halves are then exchanged with probability of a half. Next, $[0,1/2)$ is further divided into two halves of length $1/4$, which are then permuted randomly. Similarly, the right half $[1/2,1)$ is also divided into two halves and permuted randomly. This procedure is continued recursively until machine precision is reached. In higher dimension, the scrambling is conducted independently for each dimension.}
Using a scrambled $(t,d)$-sequence, the variance is $o(n^{-1})$ for $L_2$ integrable functions, which has infinite asymptotic efficiency compared with MC. If the integrand is sufficiently smooth, the variance has order $O(n^{-3}(\log n)^{d-1})$ \cite{smoovar,localanti}.
For a comprehensive introduction on QMC and RQMC, see \cite{dick:pill:2010} and \cite{dick:kuo:sloa:2013}.


Despite the usual dimension-dependent error bounds, QMC and RQMC methods work well empirically for many high-dimensional integrals. One explanation is that some integrands might have a low ``effective" dimension, even though the nominal dimension is high.
There are several definitions of effective dimension and mean dimension using the ANOVA decomposition of a function. In the ANOVA decomposition $f(\bfx)=\sum_{u\subseteq\blue{1:d}}f_u(\bfx)$, $f_u$ depends on $\bfx$ only through $\bfx_{u}$, $\int_0^1f_u(\bfx)\rd x_j=0$ if $j\in u$, and $\int_{[0,1]^d}f_u(\bfx) f_v(\bfx)\rd \bfx$ if $u\neq v$. This decomposition gives a decomposition of the variance of $f$:
\begin{align*}
\Var{f}=\sum_{u\subseteq\blue{1:d}}\sigma_u^2,\quad \text{ where }\sigma_u^2=\begin{cases}
     \int_{[0,1]^d}f_u(\bfx)^2\rd \bfx, & u\neq\emptyset\\
     0, & u=\emptyset
\end{cases}.
\end{align*} 
If $\{\bfx_1,\ldots,\bfx_n\}$ is an RQMC point set, the RQMC variance has the expression
\begin{align*}
\Var{\frac1n\sum_{i=1}^n f(\bfx_i)}=\frac1n\sum_{u\neq\emptyset}\Gamma_u\sigma_u^2,     
\end{align*}
where $\Gamma_u$ are called gain coefficients \cite{snetvar}. Usually, $\Gamma_u\ll 1$ for subsets $u$ with small cardinality. So if the variance of $f$ is dominated by $\sigma_u^2$ with small $u$, RQMC variance can be much smaller than MC variance. 
The effective dimension or mean dimension \blue{measures to which extent} the integrand is dominated by ANOVA components of small cardinality. For the explicit expressions of the ANOVA decomposition, effective dimension, and mean dimension, we refer to \cite{mcbook}. 

For Gaussian integrals $\EE[\bfx\sim\N(0,I_d)]{f(\bfx)}$, one can reduce the effective dimension of the integrand by finding an orthogonal matrix $U$ and replace the integrand $f(\bfx)$ by $f_U(\bfx)=f(U\bfx)$. The choices of $U$ are often motivated by heuristics that aim to minimize the effective dimension of some approximations of $f$. For example, the LT method approximates the integrand by its first-order Taylor expansion at some fixed points. 
The active subspaces method is also a dimension reduction method known as gradient PCA (GPCA) \cite{xiao2019enhancing}. 
The reason for choosing $U$ as the eigenvectors of the matrix $C$ defined in Equation~\eqref{def: C} is that QMC point sets tend to be more evenly spaced in the first few dimensions, and aligning the directions with higher variance with the first few dimensions of QMC points can potentially reduce the effective dimension.

The aforementioned dimension reduction methods are primarily designed for the Gaussian distribution which is rotationally invariant. For other distributions, one can first map the points to a Gaussian space, apply dimension reduction, and then map the points back to the original space. For example, when computing the expectation $\EE{f(\bfy)}$ where each coordinates of $\bfy$ are i.i.d. and have cumulative distribution function $F$, we have
\begin{align*}
     \EE{f(\bfy)}=\EE{f(F^{-1}(\Phi(\bfz)) ) }, \text{ where }\bfz\sim\N(0,I_d).
\end{align*} 
Any aforementioned dimension reduction methods can be applied to the function $f\circ F^{-1} \circ \Phi$.
Imai nand Tan \cite{imai2009accelerating} use this approach for derivative pricing under L\'evy models.

\subsection{CMC with active subspaces}

Conditional Monte Carlo (CMC), also known as pre-integration or Rao-Blackwellization \blue{\cite{blackwell1947conditional,casella1996rao}}, is a technique for computing the integral $\EE{f(\bfx)}$ by integrating over some variables with a closed form or a high-precision quadrature rule. For example, if we have a closed form for computing $\tf(\bfx_{-1}):=\EE{f(\bfx)\mid \bfx_{-1}}$, we can sample $\bfx_i\iid\N(0,I_{d-1})$ and estimate the integral by $n^{-1}\sum_{i=1}^n \tf(\bfx_i)$. This estimator is unbiased and has strictly smaller variance unless $f$ does not depend on $x_1$ at all. Moreover, under some conditions, pre-integration improves the smoothness of the integrand and makes it QMC-friendly \cite{grie:kuo:leov:sloa:2018}.

Recently, Liu and Owen \cite{liu2022pre} propose\blue{d} to pre-integrate over a linear combination of the variables, instead of one of the input variables. Equivalently, they propose to find an orthogonal matrix $U$ and pre-integrate over the first variable of the rotated integrand $f_U(\bfx)=f(U\bfx)$. The orthogonal matrix $U$ is found by the active subspace \cite{cons:2015}, which consists of the eigenvectors (in descending order) of the matrix $C$ defined in Equation~\eqref{def: C}.
The subspace spanned by the top eigenvectors of $C$ is known as the active subspace of $f$. Intuitively, the top eigenvectors of $C$ correspond to the direction where $f$ has the greatest variation on average. In particular, the first eigenvector of $C$ is approximately the direction where $f$ has the highest variance. Hence, pre-integrating over that direction can potentially reduce the variance by a large amount.


\subsection{Derivative pricing}
\label{sec: option pricing}
Monte Carlo and QMC methods are widely used for derivative pricing. In financial engineering, an asset price $S_t$ is often modeled by the stochastic differential equation (SDE)
\begin{align}
     \rd S_t = r S_t \rd t + \sigma S_t \rd W_t,
     \label{equ: GBM}
\end{align}
where $r$ is the risk-neutral interest rate, $\sigma$ is the volatility, and $W$ is a standard Brownian motion.
For this simple Black-Scholes model, $S_t$ has a closed-form solution
\[
S_t=S_0\Exp{(r-\sigma^2/2)t + \sigma W_t}.     
\]
For an Asian call option with strike $K$ and maturity $T$, the price is given by
\begin{align*}
     e^{-rT}\EE{\left(\bar S - K\right)_+},\quad \text{where }\bar S=\frac1T\int_0^T S_t \rd t.
\end{align*}
The expectation is estimated by simulating paths of $S_t$ at a discrete time grid and averaging over the value of $(\bar S - K)_+$.
Specifically, we divide the interval $[0,T]$ into $d$ equal-length intervals. Denote $\Delta t=T/d$ and $t_j=j\Delta t$, and $S_j=S_{t_j}$. Then $\bar S$ is approximated
\begin{align*}
     \frac1d\sum_{j=1}^d S_0 \Exp{(r-\sigma^2/2)j\Delta t + \sigma B_j },
\end{align*}
where $B=(B_1,\ldots,B_d)\tran\sim\N(0,\Sigma)$ and $\Sigma_{ij}=\Delta t \cdot (i\wedge j)$. We refer to $B$ as the discrete Brownian motion.


The discrete Brownian motion $B$ can be simulated by choosing a square-root of $\Sigma$. If $R$ satisfies $RR\tran=\Sigma$, we can generate $\bfz_i\iid\N(0,I_d)$ and let $B_i=R \bfz_i\sim\N(0,\Sigma)$. As mentioned before, $R$ can be chosen by the standard construction, PCA construction, or Brownian bridge construction.
For the Asian options, it is observed that the PCA construction is superior to the standard construction when using RQMC. This is because the integrand has a lower effective dimension under the PCA construction.

Further variance reduction can be achieved by conditional Monte Carlo. With the choice of $R$, the integrand is $f_R(\bfz):=\left(\bar S(R, \bfz) - K \right)_+$ where
\begin{align*}
\bar S(R,\bfs)=\frac{S_0}{d}\sum_{j=1}^d \Exp{(r-\sigma^2/2) j\Delta t + \sigma \sum_{k=1}^d R_{jk}z_k },\quad \bfz\sim\N(0,I_d).
\end{align*}
If $R_{\cdot 1}\geq 0$ \blue{componentwise}, then $\bar S(R,\bfz)$ is monotonically increasing in $z_1$. Thus, for each $\bfz_{2:d}$, there exists a unique $\gamma =\gamma(\bfz_{2:d})$ such that $\bar S(R,(\gamma,\bfz_{2:d}) )=K$. Given the threshold $\gamma$, which can be found by the Newton method, the conditional expectation of $f_R(\bfz)$ given $\bfz_{2:d}$ has the closed form
\begin{align}
 \tf(\bfz)=\frac{S_0}{d}\sum_{j=1}^d \Exp{(r-\sigma^2/2) j\Delta t +\sigma\sum_{k=2}^d R_{jk} z_k + \sigma^2 R_{j1}^2/2 } \bar\Phi(\gamma-\sigma R_{j1})-K\bar\Phi(\gamma),
 \label{equ: preint formula}
\end{align}
where $\bar\Phi=1-\Phi$. \blue{This function, when composed with the inverse CDF $\Phi^{-1}$, has singularities at the boundaries of the unit cube, resulting in unbounded Hardy-Krause variation. Owen \cite{owen2006halton} introduced a boundary growth condition, under which RQMC points achieve an error rate of $O(n^{-1+\ep})$ for any $\ep>0$. Subsequently, He \cite{he2019error} showed that the pre-integrated function $\tf$ in Equation~\eqref{equ: preint formula}, along with some Greeks, satisfies the boundary growth condition introduced in \cite{owen2006halton}. As a result, when using the first $n$ points of a scrambled $(t,d-1)$-sequence to integrate $\tf\circ\Phi^{-1}$, the error rate is $O(n^{-1+\ep})$. Moreover, He \cite{he2019error} also showed that if $R_{\cdot 1}$ has mixed signs, then the pre-integrated function is not infinitely smooth and the error rate of $O(n^{-1+\ep})$ is not obtained. This underscores the importance of the monotonicity condition $R_{\cdot1}\geq 1$ in facilitating both tractable computation for pre-integration and improved error rates.
}


However, there are more complex integrands than $(\bar S-K)_+$ and more intricate models than the constant-volatility model.
In such cases, pre-integration might not have a closed form for arbitrary choices of $R$. The rest of the paper will focus on addressing this issue.

\section{Active subspaces with constraints}
\label{sec: AS constraints}

As discussed before, to efficiently apply pre-integration with RQMC to compute a Gaussian integral $\EE[\N(0,I_d)]{f(\bfx)}$, we need to choose an orthogonal matrix $U$ and then pre-integrate $x_1$ of the rotated integrand $f_U(\bfx)=f(U\bfx)$.
The choice of the orthogonal matrix $U$ can be divided into two steps:
\begin{enumerate}
\item Choosing the first column $U_1$ of $U$ so that $f_U$ has big variation along $x_1$ and $\tf_U(\bfx_{2:d})=\EE{f(\bfx)\mid \bfx_{2:d}}$ is easy to compute.
\item Choosing the remaining $d-1$ columns of $U$ so that $\tf_U$ has low effective dimension.
\end{enumerate}
The choice of $U_1$ in step (1) is problem specific and the strategies vary with integrands. We will discuss how to choose $U_1$ for some specific examples in Section \ref{sec: examples}. The rest of this section will be focused on step (2). That is, if $U_1$ has been determined already, how do we choose the remaining columns of $U$ so that the resulting integrand $\tf_U(\bfx_{2:d})$, which will be integrated with RQMC, has a low effective dimension.

Now we fix the first column of $U$ to be $U_1$. To determine the second column of $U$, following the idea of active subspaces, we want to maximize the explained variance along that direction subject to the orthogonality constraint:
\begin{align*}
     \text{maximize}_{v\in\R^d}\quad & v\tran Cv, \\
     \text{s.t.}\quad &v\tran U_1=0, v\tran v=1
\end{align*}
where $C$ is defined in~\eqref{def: C}.
To determine the third column, we need the constraint that it is orthogonal to the first two columns of $U$.
If we have the first $(k-1)$ columns $U_1,\ldots,U_{k-1}$ determined, we choose the $k$-th column by solving
\begin{align*}
     \text{maximize}_{v\in\R^d}\quad &v\tran Cv,\\
     \text{s.t.}\quad & v\tran U_i=0,\;\forall \; 1\leq i\leq k-1,\\
     & v\tran v=1.
     \numberthis\label{optim}
\end{align*}
Fortunately, instead of solving these optimization problems one-by-one, we can find the optimal $U_2,\ldots,U_d$ by one eigendecomposition in the orthogonal complement of $U_1$. 
\begin{theorem}
Let $V=(v_2,\ldots,v_d)$ be an orthogonal basis spanning $U_1^\perp$, the orthogonal complement of $U_1$. \blue{Let $C$ be a positive definite matrix.} If $\widetilde C=V\tran  C V$ has eigenvectors \blue{sorted in the columns of the matrix $\widetilde U$ with associated eigenvalues in descending order}, 
then the solution of the optimization problems~\ref{optim} is given by $(U_2,\ldots,U_d)=V\widetilde U$.
\end{theorem}
\begin{proof}
Because $U_j$ ($j=2,\ldots,d$) are all orthogonal to $U_1$, they can be linearly represented by the basis $(v_2,\ldots,v_d)$. That is, there exists an orthogonal matrix $\widetilde U\in\R^{(d-1)\times(d-1)}$ such that $U_{\cdot,2:d}=V\widetilde U$. Thus, solving for $U_{\cdot,2:d}$ is equivalent to solving for $\widetilde U$.
Note that the objective function for finding $U_j$ is $U_j\tran CU_j=\widetilde U_j\tran V\tran C V \widetilde U_j=\widetilde U_j\tran\widetilde C \widetilde U_j$. So the optimizing $\widetilde U_j$ is the $j$-th eigenvector of $\widetilde C$.
Therefore, the optimizing $(U_2,\ldots,U_d)$ is given by $V\widetilde U$. 
\end{proof}

The orthogonal basis $(v_2,\ldots,v_d)$ can be chosen to be the second through the last columns of the Householder matrix 
\begin{align}
     H=I_d-2ww\tran,\quad\text{where } w=\frac{U_1-e_1}{\|U_1-e_1\|},\quad e_1=(1,0,0,\ldots,0)\tran.     
     \label{equ: Householder}
\end{align}
This procedure, called the constrained active subspace (CAS), is summarized in Algorithm~\ref{algo}. In practice, the matrix $C$ is approximated by a sample average \blue{$\widehat C$}.
The algorithm can be naturally generalized when there are multiple predetermined columns of the orthogonal matrix $U$.

\begin{algorithm}[t]
     \caption{Constrained active subspaces (CAS)}
     \label{algo}
     \begin{algorithmic}
     \REQUIRE{Number of samples $M$ to compute $\widehat C$, a unit vector $U_1$}
     \ENSURE{An orthogonal matrix $U$ whose first column is $U_1$}
     \STATE Take $\bfx_0,\ldots,\bfx_{M-1}\sim\N(0,I_{d})$ by RQMC.
     
     \STATE Compute $\widehat C=\frac{1}{M}\sum_{i=0}^{M-1}\nabla f(\bfx_i)\nabla f(\bfx_i)\tran $.
     
     
     
     \STATE Let $w=\frac{U_1-e_1}{\|U_1-e_1\|}$ and $H=I_{d}-2ww\tran$. Let $V=H_{\cdot,2:d}$.
     
     \STATE Compute the eigenvectors $\widetilde U$ of $\widetilde C=V\tran \widehat CV$. 
     
     \STATE Let $U_{\cdot,2:d}=V \widetilde U$.
     
     \end{algorithmic}
\end{algorithm}

We note that the CAS algorithm differs from the previous method which finds the active subspace of the pre-integrated function $\tf(\bfx_{2:d})=\EE{f(\bfx)\mid \bfx_{2:d}}$. \blue{First and foremost, with CAS, we pre-integrate $U_1\tran\bfx$ instead of one of the coordinate variable $x_1$. The choice of $U_1\tran\mathbf{x}$ is advantageous as it is likely to capture much more variance in $f$ than $x_1$ does, resulting in a larger variance reduction. Second, the CAS method only involves computing gradients for the original function $f(\bfx)$, while the latter requires computing the gradients of $\tf(\bfx_{2:d})$, which might not be readily available in practice. Leveraging these advantages, the proposed method can achieve a greater variance reduction and offer a more efficient implementation.
}

\section{Choosing pre-integration directions}
\label{sec: examples}

The previous section talks about how to choose the second through the last columns of the rotation matrix $U$ when the first column is given. In this section, we describe how to choose the first column for some specific examples in derivative pricing.

\subsection{Stochastic volatility model}
\label{sec: SV}

The stochastic volatility model assumes that $S_t$ satisfies
\[
\frac{\rd S(t)}{S(t)}=r \rd t + \sqrt{V(t)} \rd \blue{W^{(1)}}(t), 
\]
where $V(t)$ is a stochastic process satisfying the SDE
\[
\rd V(t) = a(V(t)) \rd t + b(V(t)) \rd \blue{W^{(2)}}(t).
\]
Here, $W_1,W_2$ are two standard Brownian motions with $\Corr{\blue{W^{(1)}}(t),\blue{W^{(2)}}(t)}=\rho$, \blue{and $a$, $b$ are some functions of $V$.}
Using the same discretization, we can simulate $S_t$ and $V_t$ simultaneously using two correlated discrete Brownian motions \blue{$B^{(1)},B^{(2)}$} with $\Cov{\blue{B^{(1)},B^{(2)}}}=\rho\Sigma$, $\Var{\blue{B^{(1)}}}=\Var{\blue{B^{(2)}}}=\Sigma$.
To do this, we first generate $\bfz=(\bfz_1\tran,\bfz_2\tran)\tran\sim\N(0,I_{2d})$ and let
\begin{align}
\blue{B^{(1)}}=\sqrt{1-\rho^2}  R\bfz_1 + \rho R \bfz_2,\quad \blue{B^{(2)}}=R\bfz_2
\label{equ: gen two BM}
\end{align}
for any $R$ such that $RR\tran=\Sigma$. We let $R=R_{\text{std}}$, the standard construction matrix with $R_{ij}=\sqrt{\Delta t} \Indc{i\geq j}$. Because we will find a rotation $U$ and replace $\bfz$ by $U\bfz$, this does not rule out any possible choice of $R$.
With this choice of $R$, $\blue{B^{(2)}_{j+1}-B^{(2)}_{j}}=\sqrt{\Delta t}z_{2,j}$, $\blue{B^{(1)}_{j+1}-B^{(1)}_{j}}=\sqrt{\Delta t}(\sqrt{1-\rho^2} z_{1,j}+\rho z_{2,j})$. We can simulate $S$ (in the log space) and $V$ by the Euler-Maruyama scheme
\begin{align*}
     \log S_{j+1}&=\log S_j + (r-V_{j}/2)\Delta t+\sqrt{V_{j} \Delta t}(\sqrt{1-\rho^2} z_{1,j+1} + \rho z_{2,j+1}),\\ \numberthis\label{equ: sv euler}
     V_{j+1}&=V_{j} + a(V_{j})\Delta t + b(V_{j}) \sqrt{\Delta t}z_{2,j+1},\; \text{ for } j=0,\ldots,d-1.
\end{align*}
The integrand is still $f(\bfz)=(d^{-1}\sum_{j=1}^d S_j - K )_+$ for the Asian call option.
Pre-integrating $z_{1,1}$ is still direct by noting that $V_j$'s do not depend on $z_{1,1}$ and $S_j$'s
 are monotone in $z_{1,1}$. 

However, in this integrand, pre-integrating $z_{1,1}$ might not reduce the variance by a lot. \blue{According to our numerical experiments, pre-integrating $z_{1,1}$ in these examples only reduce the error by factors ranging from 4 to 14.}
In order to reduce as much variance as we can by pre-integration, we first rotate the integrand by an orthogonal matrix $U$ so that $z_{1,1}$ is a more important direction in the rotated integrand $f(U\bfz)$. 
However, because both $V_j$ and $S_j$ depend on $\bfz$, pre-integration might not be possible if we rotate $\bfz$ by an arbitrary orthogonal matrix $U$. 
In fact, for pre-integration to be easily carried out, the only requirement is that $V_j$'s should be independent of $z_{1,1}$ and that $S_j$'s should be monotone in $z_{1,1}$. This is true as long as $U_{1:d,1}\geq 0$ and $U_{d+1:2d,1}=0$. 

Therefore, using the idea of active subspaces, we choose $U_{\cdot 1}$ to be the solution of the optimization problem
\begin{align*}
\text{maximize}_{v\in\R^{2d}}\quad& v\tran Cv,\\
\text{s.t.}\quad&v\tran v=1, v_{d+1:2d}=0.
\numberthis\label{equ: optim SV}
\end{align*}
Here, $C$ is the $2d\times 2d$ matrix defined in Equation~\eqref{def: C} for the integrand $f(\bfz)=(d^{-1}\sum_{j=1}^d S_j-K)_+$ with $S_j$ given by Equation~\eqref{equ: sv euler}. Without the constraint $v_{d+1:2d}=0$, the solution is the first eigenvector of $C$, i.e. the first direction in the usual active subspace. With the constraint, the optimization problem above finds the direction in which $f$ has the greatest average variation among those that satisfy $v_{d+1:2d}=0$. 
If we partition $C$ into four $d\times d$ blocks
\begin{align}
C=\begin{pmatrix}
C_{11} & C_{12} \\ C_{21} & C_{22}
\end{pmatrix},
\label{equ: partition C}
\end{align}
the solution to problem \eqref{equ: optim SV} is such that $v_{1:d}$ is the first eigenvector of $C_{11}$ and $v_{d+1:2d}=0$. After choosing $U_1$ to be the opimizing $v$, we can apply Algorithm~\ref{algo} to determine the remaining $(2d-1)$ columns of $U$, as described in Section~\ref{sec: AS constraints}. Note that in this example, the first eigenvector of $C_{11}$ always has non-negative signs. This is because $f(\bfz)$ is non-decreasing in all $\bfz_{1:d}$, thus all entries of $C_{11}$ are non-negative, leading to the first eigenvector to have non-negative entries. 



%

Now that we have found the rotation matrix $U$, we describe how to pre-integrate $z_{1,1}$ in $f(U\bfz)$.
We first replace $\bfz$ by $\bfx$ with
\[
     \begin{pmatrix} \bfx_1 \\ \bfx_2 \end{pmatrix} = U \begin{pmatrix} \bfz_1 \\ \bfz_2 \end{pmatrix}.
\]
Following Equation~\eqref{equ: sv euler}, we can write
\begin{align*}
S_j&=S_0\Exp{\sum_{k=1}^{j}(r- V_{k-1}/2)\Delta t + \sqrt{\Delta t}\sum_{k=1}^{j}\sqrt{V_{k-1}} (\rho x_{2,k} + \sqrt{1-\rho^2} x_{1,k}) }\\
&=\exp(c_j z_{1,1} ) \cdot \zeta_j(\bfz_{1,2:d},\bfz_2),
\end{align*}
where
\[
c_j=\sqrt{\Delta t} \sum_{k=1}^{j}\sqrt{V_{k-1}} \sqrt{1-\rho^2} U_{k,1}
\]
and 
\begin{align*}
&\zeta_j(\bfz_{1,2:d},\bfz_2)=S_0\Exp{\sum_{k=1}^{j}(r-V_{k-1}/2)\Delta t + \sqrt{\Delta t}\sum_{k=1}^{j}\sqrt{V_{k-1}} \cdot \right.\\
&\qquad\qquad\;\left.[\rho \sum_{l=1}^d (U_{k+d, l}z_{1,l} + U_{k+d,l+d}z_{2,l}) +
 \sqrt{1-\rho^2} (\sum_{l=2}^d U_{k,l}z_{1,l}+ \sum_{l=1}^d  U_{k,l+d}z_{2,l})]
}.
\end{align*}
Both $c_i$ and $\zeta_i$ do not depend on $z_{1,1}$ since $U_{d+1:2d,1}=0$. Because $U_{1:d,1}\geq 0$, $c_j\geq 0$ for all $j$, for each $(\bfz_{1,2:d},\bfz_2)$, there exists a unique $\gamma=\gamma(\bfz_{1,2:d},\bfz_2)$ such that when $z_{1,1}=\gamma$, it holds that
$\bar S=K$.
Then the expectation of $(\bar S-K)_+$ conditional on $(\bfz_{1,2:d},\bfz_2)$ is equal to
\begin{align*}
\EE{(\bar S-K)_+\mid (\bfz_{1,2:d},\bfz_2)}
&=\frac{1}{d}\sum_{j=1}^d \zeta_j e^{c_j^2/2}\bar\Phi(\gamma -c_j) - K\bar\Phi(\gamma).
\end{align*}
The threshold $\gamma$ can be found by a root-finding algorithm similarly as in the constant-volatility case. \blue{We dropped the dependence of $\zeta_j$ and $\gamma$ on $(\bfz_{1,2:d},\bfz_2)$ for notational simplicity.
}

\subsection{Basket option}
\label{sec: spread}

\blue{An option might involve multiple underlying assets and the payoff depends on a weighted average of them, i.e.
\begin{align*}
\EE{\left(\sum_{\ell=1}^L w_\ell \bar S^{(\ell)} -K\right)_+},
\end{align*}
where $w_\ell$ can be either positive or negative. We assume the price $S^{(\ell)}$ satisfies the SDE~\eqref{equ: GBM} driven by the Brownian motion $W^{(\ell)}$ with volatility $\sigma_\ell$. The correlation $\rho_{k,\ell}$ among $W^{(k)}$ and $W^{(\ell)}$ is allowed to be arbitrary as long as the matrix $(\rho_{k,\ell})_{1\leq k,\ell\leq L}$ remains positive semi-definite. This general formulation encompasses various financial instruments, including the spread option, whose payoff $\EE{(\bar S^{(1)} - \bar S^{(2)} - K)_+}$ depends on the difference between two stocks.}

\blue{Similarly as before, we let $B^{(\ell)}\in\R^d$ be the discretized Brownian motion and let $\widetilde B^{(\ell)}=\sigma_\ell B^{(\ell)}$ be the scaled version. Then $\bfB=(\widetilde B^{(1)},\ldots,\widetilde B^{(L)})\in\R^{dL}$ is a Gaussian vector with mean zero and covariance $\Lambda$, whose $(k,\ell)$-th block is $\rho_{k,\ell}\sigma_k\sigma_\ell\Sigma$ when partitioned into submatrices of size $d\times d$. Take $R$ such that $RR\tran=\Lambda$ and let $\bfB=R\bfz$ where $\bfz\sim\N(0,I_{dL})$. Then
\begin{align*}
\bar S(\bfz)=\sum_{\ell=1}^L w_\ell \bar S^{(\ell)}=\sum_{\ell=1}^L \frac{w_\ell  S^{(\ell)}_0 }{d}\sum_{j=1}^d \exp((r-\sigma_\ell^2/2)j\Delta t + \tB_{j}^{(\ell)} ).
\end{align*}
If for any $j\in 1:dL$, $w_{\ceil {j/d} } R_{j1}\geq 0$, then the above display is monotonically increasing in $z_1$. So for any $\bfz_{-1}$, if there exists $z_1^*=z_1^*(\bfz_{-1})$ such that $\bar S(z_1^*,\bfz_{-1})=K$, then it is unique and the conditional expectation has the expression
\begin{align*}
\numberthis\label{equ: spread preint}
&\EE{(\bar S(\bfz) -K)_+\mid \bfz_{-1}}=\sum_{\ell=1}^L  \frac{w_\ell S^{(\ell)}_0 }{d} \sum_{j=1}^d   \\
&\Exp{(r-\frac{\sigma_\ell^2}{2})j\Delta t + \bfz_{-1}\tran R_{(\ell-1)d+j,-1} + \frac{R_{(\ell-1)d+j,1}^2}{2}  }\bar\Phi(z_1^* - R_{(\ell-1)d+j,1} ) - K\bar\Phi(z_1^*).
\end{align*}
If such a $z_1^*$ does not exist, the formula of the conditional expectation is the same as above with $z_1^*=-\infty$.
}

\blue{We apply the CAS method to choose $R$ such that $RR\tran=\Lambda$ and satisfies the constraints $R_{\cdot1}\in\calA:=\{v\in\R^{dL}\mid  w_{\ceil {j/d} } v_j \geq 0\; \forall j\in 1:dL\}$. Moreover, we want $z_1$ to be an important direction under this choice of $R$. Starting from the Cholesky decomposition $R_0$ of $\Lambda$, we compute the matrix $\widehat C$ for the payoff function $(\bar S(\bfz)-K)_+$. Then we solve the problem
\begin{align*}
\text{maximize}_{v\in\R^{dL}}&\quad v\tran\widehat C v\\
\text{s.t.}&\quad \|v\|_2=1,\, R_0 v\in \calA.
\end{align*}
In practice, we compute the first eigenvector $v_1$ of $\widehat C$ and let $r_1=R_0 v_1$. We then project $r_1$ into $\calA$. Specifically, let $\calI=\{j\in1:dL\mid r_{1,j} w_{\ceil{j/d}}\geq 0 \}$ be the index set of $r_1$ with the correct signs. Note that we can always flip the sign of the entire vector $r_1$ so that $\|r_{1,\calI}\|_2\geq \|r_{1,-\calI}\|_2$. Then we set $r_{1,-\calI}$ to 0 and let $U_1=R_0^{-1}r_1/\|R_0^{-1}r_1\|_2$ be the first column of the rotation matrix of $U$. This truncation step ensures that the first column of $R_0U$ satisfies the constraint of being in $\calA$.
Once $U_1$ is obtained, we invoke Algorithm~\ref{algo} to compute the remaining $dL-1$ columns of $U$. In the end, we use the matrix $R=R_0U$ in the expression~\eqref{equ: spread preint}.
}

\subsection{Greeks}
\label{sec: greeks}

Greeks are derivatives of payoff functions with respect to some parameters like interest rate or volatility. 
They often take the form of $f(\bfx)=g(\bfx)\Indc{h(\bfx)\geq K}$, which are discontinuous. The active subspaces method is not well defined because the gradient does not exist at the discontinuous points. Although the payoff functions considered before are not differentiable at the threshold where $\bar S=K$, they are continuous so ignoring a measure-zero set of non-differentiable points does not cause problems. But for functions with jumps, ignoring the non-differentiable points also ignores the variation caused by the jumps, which might be crucial in determining the active subspaces.

Fortunately, the Greeks can be transformed to continuous functions because they share the separation-of-variable property \cite{wang2016handling}.
Suppose the indicator function can be written as $\Indc{h(\bfx)\geq K}=\Indc{x_1\in[\alpha(\bfx_{-1}), \beta(\bfx_{-1}) ] }$, where $\alpha$ and $\beta$ are functions of $\bfx_{-1}$ and $-\infty\leq \alpha\leq\beta\leq+\infty$.
\blue{We define
\[
T(x_1;\bfx_{-1})=\Phi^{-1}(\Phi(\alpha(\bfx_{-1}) ) + (\Phi(\beta(\bfx_{-1}) ) - \Phi(\alpha(\bfx_{-1}) )) \Phi(x_1) )
\]
and consider the transformation $\bfx\mapsto (T(x_1;\bfx_{-1}), \bfx_{-1})$. The Jacobian is equal to $|\frac{\partial T(x_1;\bfx_{-1})}{\partial x_1} |=(\Phi(\beta(\bfx_{-1}) ) - \Phi(\alpha(\bfx_{-1}) )) \frac{\varphi(x_1)}{\varphi(T(x_1;\bfx_{-1}))} $. Also note that $T(x_1;\bfx_{-1})\in[\alpha(\bfx_{-1}), \beta(\bfx_{-1}) ]$ for $x_1\in\R$.
Applying the change-of-variable formula, we have
\begin{align*}
\int_{\R^d} f(\bfx)\varphi(\bfx)\rd \bfx&=\int_{\R^{d}} g(\bfx) \Indc{x_1\in[\alpha(\bfx_{-1}, \beta(\bfx_{-1}))] } \varphi(\bfx) \rd \bfx\\
&=\int_{\R^d} (\Phi(\beta(\bfx_{-1}) ) - \Phi(\alpha(\bfx_{-1}) )) g(T(x_1;\bfx_{-1}),\bfx_{-1} )  \varphi(\bfx)\rd \bfx.
\end{align*}}
%
With this transformation, the jumps at $\alpha$ and $\beta$ are pushed to $\pm\infty$. So we end up with the continuous integrand
\[
(\Phi(\beta(\bfx_{-1})) - \Phi(\alpha(\bfx_{-1})) ) g(T(x_1;\bfx_{-1}),\bfx_{-1} ) .
\]
We can then apply the active subspaces method to this continuous integrand.

For some Greeks, there might exist a priori known important directions. For example, the \emph{gamma} for the Asian call option has the expression \cite{he2019error}
\begin{align*}
e^{-r T} \frac{\bar S}{S_0^2\sigma^2\Delta t}(\sigma\sqrt{\Delta t} x_1 - \sigma^2\Delta t ) \Indc{\bar S\geq K}.
\end{align*}
Here $x_1$ is the normal variable that generates $S_1$, i.e.
$S_1=S_0\Exp{(r-\sigma^2/2)\Delta t + \sigma \sqrt{\Delta t} x_1}$. So $x_1$ is likely to be an important direction. In our notation, $\bfx=U\bfz$. In order to make $z_1=x_1$, we can restrict the first column of $U$ to be $e_1=(1,0,\ldots,0)\tran$. Then we apply Algorithm~\ref{algo} to choose the remaining $d-1$ columns of $U$. In our numerical experiments, we see that by forcing the first column of $U$ to be $e_1$ and choosing the remaining columns by the CAS method, we can \blue{reduce the error by a factor of thousands. }. 


Moreover, it is possible that the first column of $R=R_0U$ has mixed signs, in which case the pre-integration step cannot be done by one root-finding algorithm. In fact, if $R_{\cdot1}$ has mixed signs, $\bar S=K$ can have two roots. To avoid the complexity of finding multiple roots, we replace the first column of $R$ by some vector that has all non-negative entries. This happens in some of our numerical experiments, where we find that replacing the first column of $R$ by the first column of $R_{pca}$, the PCA construction of the Brownian motion, has the best empirical performance.

\section{Experiments}
\label{sec: experiments}

In this section, we apply the proposed method to some option pricing examples and compare the error with other estimators.
For each experiment and each estimator, we use $n=2^{14}$ MC or RQMC samples and repeat 50 times. We compute the standard error of each estimators over the 50 replicates and report the error reduction factor (ERF)
\begin{align*}
     \text{ERF} := \frac{\hat\sigma^{\text{MC}} } {\hat\sigma }.
\end{align*}
The ERF is the ratio between the standard error of MC estimator and the standard error of other competing estimators using RQMC sampling.
Note that ERF is the square-root of variance reduction factor $(\hat\sigma^{\text{MC}}/\hat\sigma)^2$.

We consider two classes of estimators, one without pre-integration (RQMC) and one with pre-integration (RQMC+preint). 
For all the experiments below, the integrand involves two correlated Brownian motions. The two Brownian motions can be generated as in Equation~\eqref{equ: gen two BM}, where $R$ can be chosen by the standard construction (STD) or the PCA construction (PCA) of one Brownian motion. For the estimators without pre-integration, we can apply the active subspace method (AS) for dimension reduction. For the estimators with pre-integration, we apply the proposed constrained AS (CAS) method to choose the pre-integration direction. 

For methods involving active subspaces, we use $M=256$ gradients to estimate the matrix $C$ in \eqref{def: C} to keep the computation overhead negligible. \blue{All the} gradients are approximated by the finite difference
\begin{align*}
\left( \frac{f(\bfx+\ep e_1)-f(\bfx)}{\ep},\ldots, \frac{f(\bfx+\ep e_d)-f(\bfx)}{\ep} \right),\text{ with }\ep=10^{-6}.
\end{align*}
Because all our integrands are continuous and almost everywhere differentiable, the finite difference approximations are very close to the true gradients if they exist.

\subsection{Spread options}

We first consider the spread option $\EE{\left(\bar S^{(1)} - \bar S^{(2)} - K\right)_+}$
discussed in Section~\ref{sec: spread}. We assume $S^{(1)}$ and $S^{(2)}$ follow the SDE \eqref{equ: GBM} where the two Brownian motions have correlation $\rho$. For all experiments, the risk-neutral interest rate is set to $r=0.05$ and the dimension is $d=32$. The two volatilities in the spread option are $\sigma_1=\sigma_2=0.2$. We vary $\rho\in\{-0.5,0.5\}$ and $K\in\{-10,0,10\}$.

Table~\ref{table: spread} shows the relative efficiencies of each estimator over plain MC estimator. Several observations are in order:
\begin{enumerate}
     \item RQMC+preint with the CAS rotation achieves the smallest error in all settings. 
     \item For RQMC without pre-integration, AS is better than PCA and PCA is better than STD. This shows that dimension reduction can significantly improve RQMC errors.
     \item Pre-integration under the standard construction does not bring as much improvement as pre-integration under the PCA construction or using CAS rotation. 
     \item The relative efficiencies decrease with $K$.
\end{enumerate}

\blue{The methods involving root-finding for pre-integration entail a computation cost approximately 10-20 times higher compared to those without this step. However, the variance reduction achieved by RQMC+preint+CAS (last column) surpasses RQMC+AS (third column), the best method without pre-integration, by a factor of about 200-1600. The variance reduction achieved by pre-integration with CAS more than compensates for the additional computation cost. Moreover, the one-time computation overhead required for computing the active subspace is negligible, as it only uses $M=256$ samples. Therefore, the proposed RQMC+preint+CAS method stands out as the most efficient choice when taking into account both the variance reduction and the computation time.}

\begin{table}
\centering
\begin{tabular}{cc|ccc|ccc}
\toprule
& & \multicolumn{3}{c|}{RQMC} & \multicolumn{3}{c}{RQMC + preint} \\
$\rho$  &   $K$  & STD &  PCA &  AS &   STD &  PCA &  CAS \\
\midrule
-0.5 & -10 &         7.1 &        71.8 &      318.2 &                7.1 &              377.9 &             \textbf{7665.8}  \\
     &  0  &         5.3 &        51.2 &      228.4 &                4.9 &              274.3 &             \textbf{7471.0}  \\
     &  10 &         3.8 &        40.5 &      169.7 &                4.4 &              218.1 &             \textbf{6563.7}  \\
 0.5 & -10 &         9.9 &       125.6 &      327.7 &               12.2 &              680.9 &             \textbf{4564.6}  \\
     &  0  &         5.1 &        70.0 &      223.0 &                6.4 &              352.4 &             \textbf{4400.2}  \\
     &  10 &         3.5 &        35.5 &      112.2 &                5.0 &              243.7 &             \textbf{4511.5} \\
\bottomrule
\end{tabular}
\caption{Error reduction factor for the spread option. Plain MC estimator has a factor of 1. Each estimator uses $n=2^{14}$ scrambled Sobol' points and the standard error is computed using 50 random replicates. The methods in the first three columns apply RQMC to the integrand $f_U$ with rotation matrices $U$ chosen by standard construction (STD), PCA construction (PCA), and active subspaces (AS). The methods in the last 3 columns apply RQMC with pre-integration, with rotation matrices $U$ chosen by standard construction, PCA construction, and constrained active subspaces (CAS). The same protocol applies to the tables below.
}
\label{table: spread}
\end{table}

\subsection{Stochastic volatility models}

We consider three stochastic volatility models: the Hull-White model~\cite{hull1987pricing}, the Heston model~\cite{heston1993closed}, and the Stein-Stein model~\cite{stein1991stock}.
The Hull-White model assumes that $V(t)$ is a geometric Brownian motion satisfying
\[
\rd V(t)= \nu V(t) \rd t + \xi V(t) \rd W_2(t).
\]
We apply the Euler scheme in the log space of $V$ and simulates $V_j$ by
\begin{align*}
     \log V_{j}=\log V_{j-1} + (\nu-\xi^2/2) \Delta t+\xi \sqrt{\Delta t} z_{2,j}.
\end{align*}
We take $r=0.05,\,V_0=0.2,\,\xi=0.5,\,\nu=0$. We vary $\rho\in\{-0.5,0.5\}$ and $K\in\{90,100,110\}$. The results are in Table~\ref{table: hull white}. 

\begin{table}
\centering
\begin{tabular}{cc|ccc|ccc}
\toprule
& & \multicolumn{3}{c|}{RQMC} & \multicolumn{3}{c}{RQMC + preint} \\
$\rho$  &   $K$  & STD &  PCA &  AS &   STD &  PCA &  CAS \\
\midrule
-0.5 & 90  &          7.6 &        49.3 &       82.7 &                8.7 &               95.6 &              \textbf{278.6} \\
     & 100 &          5.3 &        33.9 &       46.9 &                6.5 &               77.1 &              \textbf{237.8} \\
     & 110 &          4.1 &        29.1 &       31.7 &                5.5 &               63.9 &              \textbf{155.3} \\
     0.5 & 90  &          6.4 &        22.9 &       23.2 &                8.1 &               27.5 &               \textbf{33.4} \\
     & 100 &          4.6 &        19.4 &       32.1 &                6.0 &               23.9 &               \textbf{70.5} \\
     & 110 &          3.6 &        16.9 &       27.1 &                4.2 &               20.3 &               \textbf{44.4} \\
\bottomrule
\end{tabular}
\caption{Error reduction factors for the payoff function under Hull-White model.}
\label{table: hull white}
\end{table}


The Heston model assumes that $V(t)$ follows a Cox-Ingersoll-Ross (CIR, square-root mean-reverting) process
\[
\rd V(t) = \kappa(\theta - V(t)) \rd t + \sigma \sqrt{V(t)} \rd W_2(t).
\]
We simulate $V_j$ by
\begin{align*}
     V_{j}=V_{j-1} + \kappa (\theta - V_{j-1})\Delta t + \sigma\sqrt{V_{j-1}\Delta t} z_{2,j}.
\end{align*}
We take $r=\sigma=0.05,\,V_0=\theta=0.2,\,\kappa=1$. Under the Feller condition $2\kappa\theta\geq\sigma^2$, the process $V$ is strictly positive.
Achtsis et al. \cite{acht:cool:nuye:2013} proposed a conditional sampling scheme for a barrier option under the Heston model. Their method uses the separation-of-variable property to remove the discontinuity caused by the barrier instead of pre-integration.
The results are in Table~\ref{table: heston}. 

\begin{table}
\centering
\begin{tabular}{cc|ccc|ccc}
\toprule
& & \multicolumn{3}{c|}{RQMC} & \multicolumn{3}{c}{RQMC + preint} \\
$\rho$  &   $K$  & STD &  PCA &  AS &   STD &  PCA &  CAS \\
\midrule
-0.5 & 90  &        12.1 &        71.6 &      134.1 &               13.2 &              244.3 &              \textbf{613.4} \\
     & 100 &         6.9 &        67.8 &       98.9 &                8.4 &              186.3 &              \textbf{411.3} \\
     & 110 &         4.7 &        49.1 &       84.7 &                6.2 &              145.2 &              \textbf{370.3} \\
 0.5 & 90  &         9.1 &        63.2 &       79.4 &               12.2 &              143.2 &              \textbf{352.1} \\
     & 100 &         5.5 &        44.5 &       74.4 &                8.0 &              117.2 &              \textbf{347.1} \\
     & 110 &         4.6 &        46.0 &       62.2 &                5.4 &               95.8 &              \textbf{254.0} \\
\bottomrule
\end{tabular}
\caption{Error reduction factors for the payoff function under Heston model.}
\label{table: heston}
\end{table}

The Stein-Stein model assumes that $V(t)$ follows the Ornstein-Uhlenbeck process
\[
\rd V(t)=\kappa(\theta-V(t))\rd t+\sigma V(t)\rd W_2(t).
\]
We take $r=0.05,\,V_0=\theta=0.2,\,\sigma=0.1,\,\kappa=1$. The results are in Table~\ref{table: stein stein}. 

\begin{table}
\centering
\begin{tabular}{cc|ccc|ccc}
\toprule
& & \multicolumn{3}{c|}{RQMC} & \multicolumn{3}{c}{RQMC + preint} \\
$\rho$  &   $K$  & STD &  PCA &  AS &   STD &  PCA &  CAS \\
\midrule
-0.5 & 90  &        12.2 &        71.9 &      135.9 &               13.3 &              245.5 &              \textbf{671.7} \\
     & 100 &         6.9 &        67.2 &       99.6 &                8.4 &              187.4 &              \textbf{481.0} \\
     & 110 &         4.7 &        49.2 &       84.4 &                6.2 &              146.0 &              \textbf{420.5} \\
 0.5 & 90  &         9.2 &        63.8 &       73.4 &               12.2 &              145.3 &              \textbf{372.6} \\
     & 100 &         5.6 &        44.7 &       64.1 &                8.1 &              119.1 &              \textbf{273.9} \\
     & 110 &         4.6 &        46.4 &       69.7 &                5.4 &               97.6 &              \textbf{259.1} \\
\bottomrule
\end{tabular}
\caption{Error reduction factors for the payoff function under Stein-Stein model.}
\label{table: stein stein}
\end{table}

In all three models, the use of pre-integration with the PCA or CAS rotations leads to significant improvements compared to methods without pre-integration. Pre-integration with the standard construction, although it has a closed form, does not bring much improvement. The CAS method with pre-integration consistently performs better than other methods.

To provide better intuition on why different rotations have very different performances, we plot the contours of the payoff function under the Heston model with different rotations in Figure~\ref{fig: heston contour}. The contours are in the first two variables $x_1$ and $x_2$, while the other coordinates $\bfx_{3:2d}$ are sampled at random. Under the standard construction (left panel), the payoff function shows no special trend in $x_1$ and $x_2$, meaning neither $x_1$ nor $x_2$ are important.
Under the AS rotation (middle panel), the function has larger variation along $x_1$ and is almost flat in $x_2$, resulting in a lower effective dimension.
We also observe that the CAS rotation (right panel), which is the AS rotation subject to the constraint that $U_{d+1:2d,1}=0$, is visually different from the AS rotation. 
Although the CAS rotation may not align $x_1$ with the most important direction of the integrand as effectively as the AS rotation, it does offer a convenient way to perform pre-integration, which is not possible with the AS rotation.
\begin{figure}
\centering
\includegraphics[width=\textwidth]{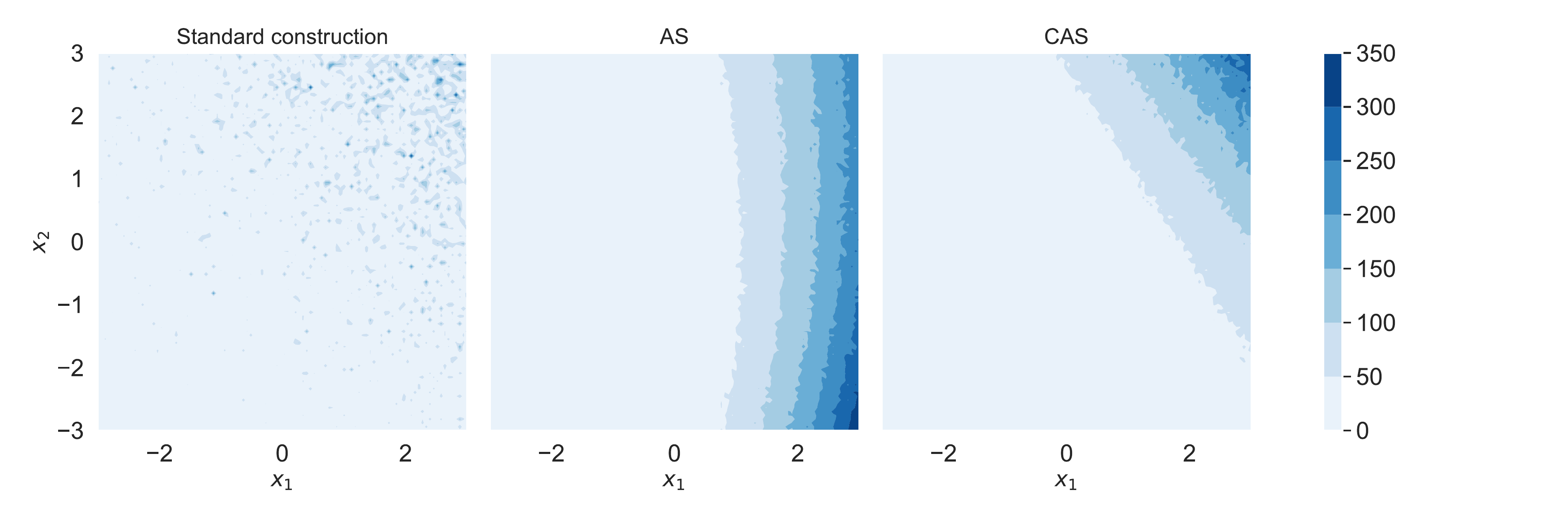}
\caption{Contours of the payoff function under the Heston model with different rotations. Left: no rotation (using the standard construction). Middle: the integrand is rotated using the orthogonal matrix found by active subspaces method. Right: the integrand is rotated by the orthogonal matrix $U$ found by CAS subject to the constraint that $U_{d+1:2d,1}=0$. This example takes $\rho=0.5$ and $K=100$.}
\label{fig: heston contour}
\end{figure}

\subsection{Greeks}

We use the pathwise estimate of the Greeks: \emph{delta}, \emph{gamma}, \emph{rho}, \emph{theta}, \emph{vega} for the Asian call option $\left(\bar S - K \right)_+$ under the constant volatility model~\eqref{equ: GBM} with $r=0.05$ and $\sigma=0.2$. We use the formulae given in \cite{he2019error}. The Greeks are discontinuous where $\bar S=K$. As discussed in Section~\ref{sec: greeks}, we first use the separation-of-variable trick to remove the discontinuities for the Greeks and then compute the active subspace for the continuous integrands. 
We start with the standard construction $R_{\text{std}}$ and find the orthogonal matrix $U$ by AS method. Then we replace $R_{\text{std}}$ by $R=R_{\text{std}} U$. 
In order for pre-integration to be feasible, we need the first column of $R$ to be all non-negative.  
In the experiments, we observe that $R_{\cdot1}$ has mixed signs for the \emph{gamma} and the \emph{rho}. So we need to choose $R_{\cdot1}$ to be some other unit vector whose entries are all non-negative and determine the rest $d-1$ columns by the CAS outlined in Algorithm~\ref{algo}.

For the \emph{rho}, we choose $R_{\cdot 1}$ to be the first column of $R_{\text{pca}}$ because the first principal component of the Brownian motion is known to be an important direction for the Asian options. 
For the \emph{gamma}, we choose $R_{\cdot 1}$ to be first column of $R_{\text{std}}$, because the first variable under the standard construction is an important variable as discussed in Section~\ref{sec: greeks}.

The results are shown in Table~\ref{table: greeks}. In each row, the greatest value as well as those that are larger than 90\% of the greatest value are marked bold. We have a few observations:
\begin{enumerate}
\item For the \emph{gamma}, the \emph{theta}, and the \emph{vega}, the proposed method has the smallest or close to smallest error among competing estimators.

\item For the \emph{gamma}, the advantage of the proposed method is particularly \blue{large}. This is the case where we take into account the specific form of the integrand. The proposed method has the flexibility to incorporate additional knowledge of the integrand, leading to significant improvements.

\item For the \emph{delta} and the \emph{rho} with $K=90$, pre-integration with PCA construction has the smallest error. These are the integrands for which PCA construction has very good performance.

\end{enumerate}

\begin{table}
\centering
\begin{tabular}{cc|ccc|ccc}
\toprule
& & \multicolumn{3}{c|}{RQMC} & \multicolumn{3}{c}{RQMC + preint} \\
Greeks  &   $K$  & STD &  PCA &  AS &   STD &  PCA &  CAS \\
\midrule
\emph{delta} & 90  &  2.0 &  6.1 &  7.7 &         3.2 &      \textbf{3685.5} &     1523.0 \\
     & 100 &  1.9 &  9.2 &  9.5 &         3.9 &      \textbf{3614.2} &     2839.1 \\
     & 110 &  1.8 &  7.8 &  7.3 &         2.9 &      \textbf{2653.3} &     1736.8 \\
\emph{gamma} & 90  &  5.1 & 14.1 & 21.7 &        10.4 &       389.1 &     \textbf{3872.1} \\
     & 100 &  2.6 &  9.1 & 13.8 &         6.5 &       368.0 &     \textbf{2377.8} \\
     & 110 &  2.5 &  7.4 & 11.2 &         6.9 &       390.3 &     \textbf{2264.1} \\
\emph{rho} & 90  &  1.5 &  4.8 &  7.4 &         2.5 &      \textbf{4716.0} &     3404.9 \\
     & 100 &  1.6 &  7.6 &  9.4 &         3.1 &      \textbf{7596.7} &     \textbf{7801.6} \\
     & 110 &  1.5 &  6.8 &  8.9 &         2.6 &      3764.3 &     \textbf{4340.7} \\
\emph{theta} & 90  &  6.2 & 23.7 & 44.3 &        10.0 &      1861.3 &     \textbf{2436.4} \\
     & 100 &  8.0 & 65.4 & 80.4 &        11.9 &      1834.8 &     \textbf{2607.4} \\
     & 110 &  3.2 & 14.6 & 13.8 &         4.7 &      \textbf{2082.1} &     \textbf{1878.9} \\
\emph{vega} & 90  &  4.3 & 15.5 & 27.1 &         7.1 &      1848.7 &     \textbf{2490.1} \\
     & 100 &  4.8 & 39.5 & 48.3 &         6.6 &      1632.3 &     \textbf{2484.7} \\
     & 110 &  4.1 & 21.7 & 22.0 &         5.8 &      1928.0 &     \textbf{2388.7} \\
\bottomrule
\end{tabular}
\caption{Error reduction factors for the Greeks under Black-Scholes model.}
\label{table: greeks}
\end{table}

For the \emph{gamma}, our method pre-integrates the first input variable under the standard construction, the same as in the ``smooth before dimension reduction" methods \cite{weng2017efficient,he2019error}. But our method differs from the previous method in the rotation for the remaining $d-1$ variables as mentioned in Section~\ref{sec: AS constraints}. 
The previous method, which applies GPCA to the pre-integrated function, achieves an error reduction factor of 4054, 3048, 2003, respectively for $K=90,100,110$ for \emph{gamma}. This performance is similar to our method with CAS. Hence, CAS can serve as an alternative dimension reduction method for pre-integrated functions.

\subsection{Conditional density estimation}

\blue{Proposed by L'Ecuyer \cite{l2022monte}, conditional density estimation (CDE) is a method for estimating the density $f(x)$ of a random variable $X$, which often represents the output of a computer experiment. This scenario arises when we can generate copies of $X$ but lack knowledge of its distribution. The density itself is the derivative of the CDF $F(x)$. While the empirical CDF, $\hat F(x)=\frac1n\sum_{i=1}^n\Indc{X_i\leq x}$, has a derivative equal to zero almost everywhere, its conditional expectation $\tilde F(x;\calG)=\EE{\hat F(x)\mid \calG}$, given some $\sigma$-field $\calG$, might be smooth. In that case, $\EE{\frac{\rd}{\rd x} \tilde F(x;\calG)}=\frac{\rd}{\rd x}\EE{\tF(x;\calG) }=f(x) $, implying that $\frac{\rd}{\rd x}\tF(x;\calG) $ serves as an unbiased estimator of the density \cite[Proposition 1]{l2022monte}. The CDE demonstrates a Mean Integrated Square Error (MISE) of order $O(n^{-1})$, which is faster than the rate of kernel density estimation (KDE) at $O(n^{-4/5})$. When combined with RQMC, CDE achieves a faster rate of $O(n^{-2+\ep})$ \cite{l2022monte}.
}

\blue{The sum of log-normal variables finds numerous applications in various scientific fields. In ecology, it is used to model animal abundance across multiple populations \cite{talis2023difficulties}. In wireless communication, the sum of log-normals is employed to model signal shadowing resulting from multiple obstructions in the propagation environment \cite{beaulieu1995estimating}. Additionally, it finds applications in physics, electronics, and economics, as cited in \cite{szyszkowicz2009limit}. However, the sum of log-normals is notoriously known to lack a closed form and it is difficult to approximate. Here we present an efficient method to estimate its density by combining CDE with the proposed CAS method.
}

\blue{When applying CDE to estimate the density of the random variable $X$, expressed as $X=\sum_{j=1}^d e^{z_i}$, where $\bfz\sim\N(\mu,\Sigma)$, the direct approach is to ``hide" (i.e., pre-integrate) one of the log-normal terms like $e^{z_1}$ \cite{asmussen2018conditional}. However, we propose to choose some linear combination of $\bfz$ to pre-integrate. Let $R$ be a matrix satisfying $RR\tran=\Sigma$, and let $\bfy\sim\N(0,I_d)$ such that $\bfz=R\bfy+\mu$. We choose $R$ such that its first column $R_{\cdot1}\geq 0$ componentwise. This allows us to write $X$ as $\sum_{j=1}^d \exp(\mu_j+\sum_{k=1}^d R_{jk}y_k )=:h(\bf{y})$, which is monotonically increasing in $y_1$. Consequently, given $\bfy_{-1}$, there exists a unique $y_1^*=y_1^*(x;\bfy_{-1})$ such that when $y_1=y_1^*$, $h(\bfy)=x$. This allows us to derive the conditional CDF as
\begin{align*}
\tF(x; \bfy_{-1} )=\bbE\left[\Indc{h(\bfy) \leq x }\mid \bfy_{-1} \right]=\Phi(y_1^*(x;\bfy_{-1}) ).
\end{align*}
Taking the derivative w.r.t. $x$ using the chain rule, we obtain
\begin{align*}
\frac{\rd }{\rd x} \tF(x; \bfy_{-1})&=\varphi(y_1^* )\cdot \frac{\rd}{\rd x }y_1^*(x,\bfy_{-1})\\
&=\varphi(y_1^*) \cdot \left[\sum_{j=1}^d R_{j1}\cdot  \exp\left(\mu_j+R_{j1}y_1^*+\sum_{k=2}^dR_{jk}y_k \right)\right]^{-1} .
\end{align*}
The last display provides an unbiased estimator of the density.
}

\blue{It remains to choose $R$ such that $RR\tran=\Sigma$, $R_{\cdot 1}\geq 0$ componentwise, and $y_1$ is an important direction using this choice of $R$. Starting from the Cholesky decomposition $R_0$ of $\Sigma$, we construct the matrix $R$ as $R=R_0U$, where $U$ is an orthogonal matrix.
To find the matrix $U$, we begin by computing the matrix $C=\EE{\nabla h(\bfy) \nabla h(\bfy)\tran}$. The first column of $U$ is determined by solving the following optimization problem:
\begin{align*}
\text{maximize}_{v}&\quad v\tran Cv\\
\text{s.t.}&\quad \|v\|_2=1,\; R_0v\geq 0.
\end{align*}
This problem can be solved using the same method described in Section~\ref{sec: spread}. Once the first column is obtained, we proceed to compute the remaining $d-1$ columns of $U$ by invoking Algorithm~\ref{algo}. 
}

\blue{We assess the performance of the proposed CDE+CAS method against the direct CDE method in simulation. We take $d=10$, $\mu=0$, and $\Sigma$ as the auto-correlation matrix with parameter $\rho$. We estimate the density on a grid of length 200 on $[0.1,50]$. In Figure~\ref{fig: lognormal}, we plot the density curves estimated by the two methods using 1024 RQMC samples. Notably, the direct CDE estimator exhibits a spiky curve, whereas the CAS estimator provides a smoother density curve.
\begin{figure}
\centering
\includegraphics[width=.5\textwidth]{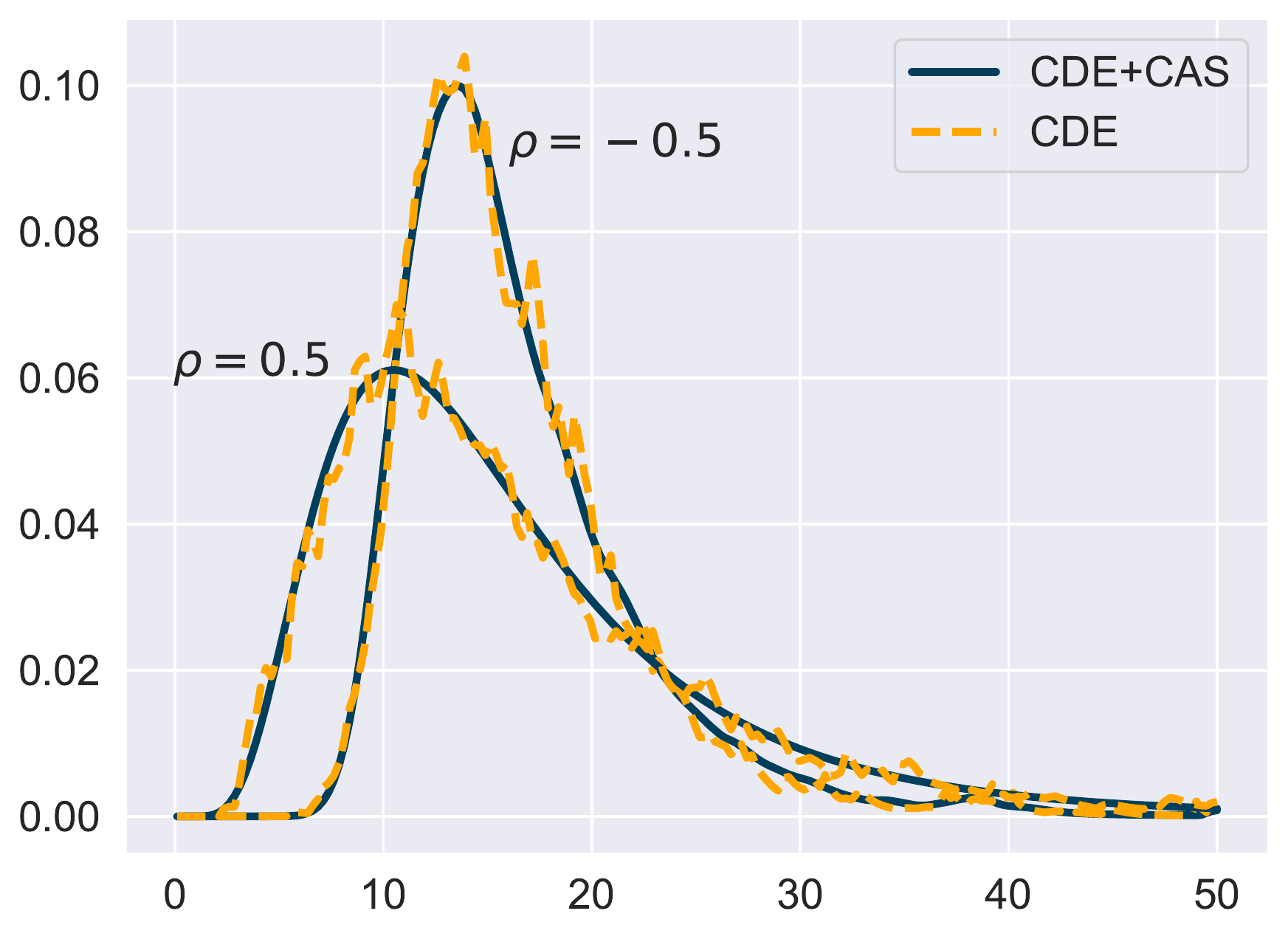}
\caption{Density of the sum of log-normal variables. The blue solid lines are estimated by the proposed CAS method, and the orange dashed lines are estimated by the direct CDE method. Both methods use $1024$ RQMC samples. Here $\rho$ is the auto-correlation parameter for the normal variables.}
\label{fig: lognormal}
\end{figure}
We also estimate the MISE following \cite{l2022monte} with 50 independent replicates. It is evident from Table~\ref{table: lognormal} that CAS significantly enhances the density estimator compared to the direct CDE method. }
\begin{table}
\centering
\begin{tabular}{ccc}
\toprule
$\rho$ &  CDE+CAS &  CDE \\
\midrule
-0.5 & 15.0 & 11.5 \\
 0.5 & 23.3 & 11.6 \\
\bottomrule
\end{tabular}
\caption{$-\log_2(\text{MISE})$ of the density estimators.}
\label{table: lognormal}
\end{table}

\subsection{Chemical reaction networks}
\blue{The proposed method can also be applied to improve the simulation of a chemical reaction network, which describes the dynamics of $N$ species participating in $J$ types of reactions.  Let $X_t\in \mathbb{R}^N$ represent the number of each species at time $t$. The main objective is to understand the distribution of $X_T$, where $T$ is the final time of interest. As an example, we consider the reversible isomerization system comprising $N=2$ species, denoted as $S_1$ and $S_2$, and $J=2$ reactions with reaction rates $c_1$ and $c_2$:
\begin{align*}
S_1 \underset{c_2}{\overset{c_1}{\leftrightarrows}} S_2.
\end{align*}
In this system, when reaction $j$ occurs, $X(t)$ is updated as $X(t)+\nu_j$, where $\nu_1=[1, -1]\tran$ and $\nu_2=[-1, 1]\tran$. The probability of reaction $j$ happening in the infinitesimal time interval $[t,t+\mathrm{d}t)$ is $a_j(X_t)\mathrm{d}t$, where $a_j(X_t)$ is known as the propensity function. The choice of $a_j(X_t)$ is often based on the Law of Mass Action \cite{beentjes2019quasi,puchhammer2021variance}. Exact simulation of the dynamics can be computationally expensive, so the dynamics are often approximated using the chemical Langevin equations (CLE), represented as the SDE
\begin{align*}
\rd X_t = (\sum_{j=1}^J \nu_j a_j(X_t) )\rd t + \sum_{j=1}^J \nu_j\sqrt{a_j(X_t)}  \rd W_{t}^{(j)}.
\end{align*}
One can simulate the SDE by applying the Euler-Maruyama discretization
\begin{align*}
X_{k+1}=X_k+\tau \sum_{j=1}^J \nu_j a_j(X_k) + \sum_{j=1}^J \nu_j \sqrt{a_j(X_k)\tau}  z^{(j)}_{k+1},\quad 0\leq k\leq d-1,
\end{align*}
where $\tau=T/d$ is the step size and $z^{(j)}_k\iid\N(0,1)$. In this case, the dimension of the problem is $d\cdot J$. For more background on the chemical reaction network, we refer to the survey \cite{higham2008modeling}.
}

\blue{
Our primary objective is to compute the expected value $\mathbb{E}[f(X_d)]$ for a given function $f$. Specifically, we consider the function $f(X_d)=\Indc{X_{d,1}\leq K}$, which corresponds to estimating the CDF of the number of species $S_1$ after time $T$. For this integrand, pre-integrating the last time step $z^{(1)}_d,z^{(2)}_d$ has a closed form, while pre-integrating any other variables does not. This is the constraint that we employ in the CAS method.
}

\blue{
Because $f$ is an indicator function, $\nabla f$ is 0 almost everywhere. To find the active subspace for the indicator function, we will instead use its smoothed version
\[
f^{\text{smo}}:=\frac12\left(1+\tanh(\frac{X_{d,1} - K}{5})\right).
\]
}
\blue{
We adopt the same parameters as in \cite{beentjes2019quasi,puchhammer2021variance}, where $X_0=[10^2,\,10^6]\tran$, $c_1=1,\,c_2=10^{-4},\,T=1.6,\,\tau=0.2$. So the integral dimension is 16. We compare four methods: RQMC, RQMC+AS, RQMC+preint, RQMC+preint+AS. In Table~\ref{table: cle}, we report the ERF of the four methods. We see that AS alone or pre-integration alone doesn't lead to significant improvement. However, their combination (preint + AS) drastically reduces the error by factors exceeding a thousand.
}
\begin{table}
\centering
\begin{tabular}{ccccc}
\toprule
Threshold $K$ &  RQMC &  RQMC + AS &  preint &  preint + AS \\
\midrule
90  &   1.6 &        9.4 &            8.2 &              1227.2 \\
100 &   2.2 &       11.1 &           10.2 &              2103.6 \\
110 &   1.4 &        9.4 &            7.4 &              1482.7 \\
\bottomrule
\end{tabular}
\caption{Error reduction factor for the chemical Langevin equation example. The first two columns do not apply pre-integration, while the last two columns pre-integrate the last time step in simulating the SDE. 
}
\label{table: cle}
\end{table}



\section{Conclusion}
\label{sec: conclusion}
This paper proposes a method called constrained active subspaces that generalizes the pre-integration with active subspaces method. The proposed algorithm takes into account both the tractability of pre-integration and the variable importance when selecting the pre-integration direction, making it more versatile and effective than previous methods. The method is demonstrated to be effective \blue{in a wide-range of applications}. It also provides a computationally-efficient alternative of conducting dimension reduction for pre-integrated functions.


%


\section*{Acknowledgments}
The author thanks Professor Art B. Owen for inspiring conversations. The author acknowledges support from the National Science Foundation grant DMS-2152780 and the Stanford Data Science Scholars program.

\bibliographystyle{apalike}
\bibliography{CAS_revision.bbl}
\end{document}